\documentclass[aps,pra,twocolumn,showpacs,floatfix,superscriptaddress]{revtex4-1}
\usepackage{color}      
\usepackage{amsmath,amsthm,thmtools,amssymb,amsfonts,bbm,graphicx,hyperref,times,subfigure}
\usepackage{mathtools}

\renewcommand{\det}{{\rm Det}\,}

\newcommand{\be}{\begin{equation}}
\newcommand{\ee}{\end{equation}}
\newcommand{\bea}{\begin{eqnarray}}
\newcommand{\eea}{\end{eqnarray}}

\renewcommand{\d}[2]{\frac{d #1}{d #2}} 

\DeclareMathOperator{\Tr}{Tr}

\def\Tr{\hbox{Tr}} 
\newtheorem{lemma}{Lemma}

\newtheorem{theorem}{Theorem}
\newcommand{\proofend}{\hfill\fbox\\\medskip }
\renewenvironment{proof}{{\noindent \bf Proof }}{$\proofend$} 

\declaretheoremstyle[notefont=\bfseries,notebraces={}{},%
    headpunct={.},postheadspace=0.5em,bodyfont=\it]{mystyle}
\declaretheorem[style=mystyle,numbered=no,name=Lemma]{thm-hand}

\graphicspath{{figures/}} 

\begin{document}
\title{The reachable set of single-mode unstable quadratic Hamiltonians}
\author{Uther Shackerley-Bennett}
\affiliation{Department of Physics and Astronomy, University College London, Gower Street, London WC1E 6BT, UK}
\author{Alexander Pitchford}
	\affiliation{Institute of Mathematics, Physics and Computer Science, Aberystwyth University}
\author{Marco G.~Genoni}
\affiliation{Quantum Technology Lab, Dipartimento di Fisica, Universit\'a degli Studi di Milano, 20133 Milano, Italy}
\affiliation{Department of Physics and Astronomy, University College London, Gower Street, London WC1E 6BT, UK}
\author{Alessio Serafini}
\affiliation{Department of Physics and Astronomy, University College London, Gower Street, London WC1E 6BT, UK}
\author{Daniel K.~Burgarth}
	\affiliation{Institute of Mathematics, Physics and Computer Science, Aberystwyth University}
\date{\today}

\begin{abstract}
The question of open-loop control in the Gaussian regime may be cast by asking which Gaussian unitary transformations are reachable by turning on and off a given set of quadratic Hamiltonians. For compact groups, including finite dimensional unitary groups, the well known Lie algebra rank criterion provides a sufficient and necessary condition for the reachable set to cover the whole group. Because of the non-compact nature of the symplectic group, which corresponds to Gaussian unitary transformations, this criterion turns out to be still necessary but not sufficient for Gaussian systems. If the control Hamiltonians are unstable, in a sense made rigorous in the main text, the peculiar situation may arise where the rank criterion is satisfied and yet not all symplectic transformations are reachable. Here, we address this situation for one degree of freedom and study the properties of the reachable set under unstable control Hamiltonians. First, we provide a partial analytical characterisation of the reachable set and prove that no orthogonal (`energy-preserving' or `passive' in the literature) symplectic operations may be reached with such controls. Then, we apply numerical optimal control algorithms to demonstrate a complete characterisation of the set in specific cases.
\end{abstract}
\maketitle

\section{Unstable quadratic controls}
Determining the set of reachable operations 
given a set of enactable quadratic Hamiltonians is an interesting endeavour, 
in view of the wide range of practical settings where such controls may be implemented --
encompassing light fields, the motional degrees of freedom of trapped ions, 
opto- and nano-mechanical oscillators and superconducting Josephson junctions,
to mention but a few -- as well as the variety of tasks they allow for, 
such as entanglement generation, squeezing, cooling and quantum communication protocols
\cite{FerraroBook,Weedbrook2012}. 
The group of all possible unitary operations 
generated by quadratic Hamiltonians corresponds to the real symplectic group,
and it is reasonable, in this context, to refer to a system as controllable if the reachable set 
coincides with the whole symplectic group.

On the theoretical side, this problem presents the peculiar feature that, although the symplectic group is a finite dimensional matrix group, the standard controllability condition provided by the Lie algebra rank criterion, which is equivalent to controllability for finite dimensional unitary groups, is not
sufficient \cite{Sachkov}. The reason for the difference is the non-compact nature of the symplectic group which provides the trajectories with the possibility of not recurring, this being impossible on compact groups. This characteristic presents itself in the operation known in quantum optics as squeezing, which can proceed indefinitely without ever recurring back to the identity operation. In a seminal paper Jurdjevic and Sussmann \cite{Jurdjevic1972} prove a sufficient condition for controllability on non-compact groups by considering the existence of control Hamiltonians that recur. 
This aspect has been noted in the context of quantum optics \cite{Genoni2012}, and an additional sufficient condition for controllability, based on the possibility of accessing a positive control Hamiltonian, has been pointed out and discussed in regard to its connection with dynamical recurrence, elsewhere referred to as `neutrality' \cite{ElliottBook}.

Furthermore it can be shown that, for systems comprised of a single degree of freedom, 
such a condition is not only sufficient but also necessary for controllability \cite{Wu2007}. One is then confronted with the intriguing situation that the Lie-algebra rank criterion may hold 
and yet the reachable set may not be the whole symplectic group. This occurs where the set of controls 
does not contain any `neutral' element, which would go back to the identity operator after a certain recurrence time.
In the quantum optical language, such a condition corresponds to the fact that all accessible Hamiltonians 
have a squeezing component which is strong enough to prevent recurrence: since all such Hamiltonians 
are unbounded from below, we shall refer to them as `unstable' \cite{paranote}.

In this paper, we apply analytical and numerical techniques 
to investigate the properties of the non-trivial reachable set 
of single-mode continuous variable systems 
under unstable quadratic controls. 
We will first review, and recast in the symplectic setting with reference to quantum optical operations, 
the analytical result of Wu, Li, Zhang and Tarn in \cite{Wu2007} concerning the uncontrollability of unstable single-mode systems. 
Part of the formalism we will introduce in the process will be instrumental in establishing 
our main analytical finding: that unstable controls do not allow one to reach 
non-trivial symplectic orthogonal transformations (corresponding to passive, 
energy-preserving phase-shifters in the lab), which implies a sharp restriction on the reachable set. 
Further, we will demonstrate the application of 
optimal control techniques to the symplectic group, first explored in \cite{wurabitz}, 
and complete the characterisation of the reachable set through the resulting numerics.

\section{Preliminaries and setup}\label{sec:setup}

Let $\hat{\bf r}=(\hat{x}_{1},\hat{p}_1,\ldots,\hat{x}_n,\hat{p}_n)^{\sf T}$ be a vector of canonical operators such that $[\hat{x}_j,\hat{p}_k]=i\delta_{jk}$. 
One refers to a Hamiltonian $\hat{H}$ as quadratic if it can be written as $\hat{H} = \frac12 \hat{\bf r}^{\sf T}H \hat{\bf r}$,
where $H$ is a real, symmetric $2n\times 2n$ matrix.
As mentioned above, unitary transformations generated by quadratic Hamiltonians correspond to transformations 
belonging to the symplectic group $\operatorname{Sp}_{2n,\mathbbm R}$ 
(this can readily be seen by considering the Heisenberg evolution of the vector of operators $\hat{\bf r}$). 
In particular, one has that the Hamiltonian evolution after a time $t$ is described by the symplectic transformation $S={\rm \,e}^{\Omega H t}$,
where
\begin{equation}
 \Omega := \bigoplus_{i=1}^n \begin{pmatrix} 0 & 1 \\ -1 & 0 \end{pmatrix}.
\end{equation}
The algebra of the symplectic group $\mathfrak{sp}_{2n,\mathbbm R}$ is therefore characterised 
as the set of matrices that can be written as $\Omega H$, where $H$ is real and symmetric.

We shall restrict ourselves here to the single-mode case ($n=1$) and consider 
the open-loop control problem defined by the evolution equation:
\begin{equation}\label{eq:system}
 \dot{S}(t) = \left( A + u(t) B \right) S(t), \quad S(0) = \mathbb{I}_2,
\end{equation}
where $u(t)$, called the control function, is any locally bounded measurable function defined on the positive time domain $[0,\infty)$ and $A, B\in \mathfrak{sp}_{2,\mathbbm R}$. 
The generator $A$ is often referred to as the drift term, in that it represents the `always-on' Hamiltonian dynamics,
while $B$ is the control term. Given $A$ and $B$, the subset $\Xi$ of $\mathfrak{sp}_{2,\mathbbm R}$ with elements of the form 
$A + vB$, $v \in \mathbbm R$, is called the set of accessible dynamical generators of the system. It is sufficient to consider a single control term even though in general we could consider, at most, three linearly independent ones. If that were the case, however, 
the drift term would be subsumed in the controls, making the problem trivial. 
Using two independent control generators, instead, would either subsume the drift term or imply that $\Xi$ contains an, as yet undefined, elliptic element. In all such cases the Lie algebra rank criterion would again become sufficient for controllability. Here, we would like to characterise systems for which this condition is satisfied but \textit{not} sufficient.

Eq.~(\ref{eq:system}) should be thought of as the equation relating a set of time-varying controls and some reachable element of the symplectic group $S(t)$.
In order to characterise which symplectic operations will be achievable for a certain accessible set $\Xi$, it is hence 
expedient to define the reachable set as follows (as customary in control theory):\smallskip

\noindent {\bf Definition:}  {\em Reachable set}.
The set of elements $S(t)$ that is reachable under Eq.~(\ref{eq:system}) for some choice of control function $u(t)$ is called the reachable set and is denoted $\mathcal{R}$.\smallskip

Following \cite{Wu2007}, it is convenient to categorise the elements $M$ of the Lie algebra $\mathfrak{sp}_{2,\mathbbm R}$ as
\begin{itemize}
\item Parabolic, if $\Tr[M^2] = 0\; ,$
\item Hyperbolic, if $\Tr[M^2] > 0 \; ,$
\item Elliptic, if $\Tr[M^2] < 0 \; .$
\end{itemize}
In the physical picture, a trajectory set by a single-mode elliptic generator is a stable one, corresponding to a strictly positive or strictly negative Hamiltonian matrix $H=\Omega M$. This fact can quickly be seen by noting that $\Tr[M^2] \propto -\det [M] \propto -\det [H]$. Such a trajectory will always recur, in the sense that 
the operator $S=\exp [Mt]$ will get 
arbitrarily close to the identity, in any matrix topology, at some positive time $t$. Note that the `neutral' elements mentioned before correspond to the elliptic elements for a single mode. Furthermore, notice that 
elliptic elements may involve a certain `amount' of squeezing, in a sense that will be made clear.
Hyperbolic generators, instead, give rise to unstable dynamics, that will never recur. Their Hamiltonian 
matrix is neither negative nor positive and does not admit a decomposition 
into normal modes, as per Williamson's theorem \cite{Williamson1936}.
Parabolic generators are instead pathological in some sense, typically they cannot be symplectically diagonalised, as
they are Jordan blocks of order $2$ (for $n=1$). A common, well known example of such generators
is the one associated with the free Hamiltonian $\hat{p}^2$. These are also non-recurring but, to the purpose of our analysis, 
it is sufficient to only consider systems where $\Xi$ solely contains hyperbolic elements, because single-mode systems involving parabolic controls either do not satisfy the Lie algebra rank criterion or are controllable.

Let us also specify a basis of $\mathfrak{sp}_{2,\mathbbm R}$: 
\begin{equation}\label{eq:basis}
 K_x = \frac{1}{2}\begin{pmatrix} 0 & 1 \\ 1 & 0 \end{pmatrix}, \quad K_y = \frac{1}{2} \begin{pmatrix} -1 & 0 \\ 0 & 1 \end{pmatrix}, \quad K_z = \frac{1}{2} \begin{pmatrix} 0 & -1 \\ 1 & 0 \end{pmatrix},
\end{equation}
which satisfies the algebra
\begin{equation} \label{eq:crs}
 [K_x,K_y] = -K_z, \quad [K_y,K_z] = K_x, \quad [K_z, K_x] = K_y.
\end{equation}

The generator $K_z$ is clearly elliptic:
the group elements ${\rm e}^{K_z t}$ are all the $2$-dimensional rotations, forming the Abelian group $\operatorname{SO}(2)$, 
which is the maximal compact subgroup of $\operatorname{Sp}_{2,\mathbbm R}$. In the lab, these correspond to 
phase-plates, or phase-shifters, that rotate the optical phase of a field: these are all the passive (or energy-preserving 
in that they preserve the number of excitations) operations belonging to $\operatorname{Sp}_{2,\mathbbm R}$. 
The hyperbolic generators $K_x$ and $K_y$ instead generate single-mode squeezing operations. 
If a linear combination of generators is considered, such as $a K_y + K_z$, with $a\in{\mathbbm R}$, it is easy to show 
that it is elliptic for $|a|<1$, parabolic for $|a|=1$ and hyperbolic for $|a|>1$. It is in this sense 
that we claimed above that elliptic generators may contain a certain `amount' of squeezing, and that for higher amounts 
the generator becomes hyperbolic.

\section{Stability and controllability}

A system evolving under Eq.~(\ref{eq:system}) is said to be controllable if the reachable set 
is the whole group $\operatorname{Sp}_{2,\mathbbm R}$.
It is known that such a system is controllable if the elements of $\Xi$ generate the Lie algebra 
({\em i.e.}, if all the algebra $\mathfrak{sp}_{2,\mathbbm R}$ 
may be obtained by their linear combinations and repeated commutations) {\em and} if 
at least one element in $\Xi$ is elliptic \cite{ElliottBook,Genoni2012}. The former condition is known as the Lie algebra rank criterion \cite{Dalessandro2007}. We shall refer to such a system as a `stable' system. 
For stable systems, the reachable set is simply the whole group. Note that all of the statements above generalise to $n$ modes.

In the case of $\operatorname{Sp}_{2,\mathbbm R}$ it is also known that a system is not controllable 
if no element in the set of accessible generators $\Xi$ is elliptic \cite{Wu2007}. We shall refer to such systems as 
`unstable'. The characterisation of the reachable set of unstable systems is a non-trivial open question, which we shall 
address here.

Before proceeding, it is useful to review the argument that establishes the 
uncontrollability of single-mode unstable systems. 
Here we will follow the treatment of \cite{Wu2007} faithfully. 
In order to make the paper self-contained, we shall reproduce the proofs of these statements 
as well in Appendix \ref{sec:proofs}.

\begin{lemma}\label{thm:symplecticsimilarity}
 If $\,\Xi$ only contains hyperbolic elements then Eq.~(\ref{eq:system}) is similar, via a symplectic transformation, to 
\begin{equation}\label{eq:esystem}
 \dot{S}(t) = (-K_x + bK_z + u(t)K_y)S(t), \quad S(0) = \mathbb{I}_2,
\end{equation}
where $b$ is some real constant with modulus strictly less than one.
\end{lemma}

Let us now denote with $\widetilde{\Xi}$ the subset of $\mathfrak{sp}_{2,\mathbbm R}$ with elements of the form 
$-K_x + bK_z + vK_y$, $v \in \mathbbm R$,
which define the accessible dynamical generators of Eq.~(\ref{eq:esystem}).
Further, we shall denote with $\widetilde{\mathcal{R}}$ 
the set of elements reachable under Eq.~(\ref{eq:esystem}) for any choice of control functions $u(t)$.

\begin{lemma}\label{lemmaf}
Any real $2 \times 2$ matrix can be written as
\begin{equation}
X = \begin{pmatrix} x_1+x_3 & x_2+x_4 \\ x_4-x_2 & x_1-x_3 \end{pmatrix},
\end{equation}
where $x_i \in {\mathbbm R}$. If $X \in \widetilde{\mathcal{R}}$ then the function
\begin{equation}
  f(x_1,x_2,x_3,x_4) := (x_1-x_4)^2 - (x_2-x_3)^2
\end{equation}
satisfies 
\begin{equation}
 f(x_1,x_2,x_3,x_4) \geq 1
\end{equation}
and
\begin{equation}
\d{}{t}f(x_1,x_2,x_3,x_4) \geq 0,
\end{equation}
for any choice of $u(t)$ in Eq.~(\ref{eq:esystem}). 
\end{lemma}

Since some symplectic transformations feature $f<1$ Lemma \ref{lemmaf} implies that such 
transformations are not reachable in systems obeying Eq.~(\ref{eq:esystem}). 
By Lemma \ref{thm:symplecticsimilarity}, this extends to all unstable systems, 
\textit{i.e.} those with only hyperbolic controls, since a symplectic similarity transformation cannot turn a non-trivial subset of 
$\operatorname{Sp}_{2,{\mathbbm R}}$ into the whole group. 
The impossibility of enacting the whole group of transformation is referred to as uncontrollability. In order to give a physical interpretation for this we first introduce singular value decomposition and then show that passive operations are unreachable.

\section{Singular value decomposition of symplectic operations} \label{sec:svd}

In order to set our findings against the backdrop of quantum optics it is 
very advantageous to introduce the singular value decomposition for 
symplectic transformations and take some care in defining its elements uniquely.
The singular value decomposition of symplectic matrices 
is often referred to as the Euler \cite{pramana} or Bloch-Messiah \cite{braunstein05} decomposition 
in the literature, 
and takes a very specific form, which is easily related to physical implementations: each symplectic 
on $n$ modes can be decomposed into the product of two passive operations, belonging to the 
intersection between $\operatorname{Sp}_{2n,{\mathbbm R}}$ and $\operatorname{SO}(2n)$,
and a direct sum of diagonal squeezing operations. In quantum optical implementations, 
passive operations correspond to beam splitters and phase-plates, which do not alter the energy of the free field.

Here, we define the singular value decomposition indicating the necessary bounds for uniqueness in one mode. This uniqueness is paramount in visualising the reachable set.\smallskip

\label{def:eulerdecomp}
\noindent {\bf Definition:}  {\em Singular value decomposition}.
Define
\begin{equation}
 \operatorname{SO}(2) :=  \left\{ \begin{pmatrix} \cos[\theta] & -\sin[\theta] \\ \sin[\theta] & \cos[\theta] \end{pmatrix} \; \middle|\; \theta \in \mathbbm R \right\} 
\end{equation}
and
\begin{equation} 
\mathcal{Z}(2,\mathbbm R) := \left\{\operatorname{diag}(1/z, z) \; \middle| \; z \in \mathbbm R, z \geq 1\right\}.
\end{equation}
Any $S \in \operatorname{Sp}_{2,\mathbbm R}$ can be decomposed as either
\begin{equation}\label{eq:svddecomposition}
 S = R_\theta Z R_\phi \quad \text{or} \quad S = R_\theta,
\end{equation}
where $R_\theta,R_\phi \in \operatorname{SO}(2)$ and $Z \in \mathcal{Z}(2,\mathbbm R)$. For the singular value decomposition to be unique, the allowed angles must be bounded such that
\begin{equation}
 -\pi+\theta_0 \leq \theta < \pi+\theta_0, \quad -\frac{\pi}{2}+\phi_0 \leq \phi < \frac{\pi}{2}+\phi_0,
\end{equation}
whre $\theta_0$ and $\phi_0$ are arbitrary but fixed. See Appendix \ref{sec:uniquesvd} for a justification of these bounds.

A generic element $S$ of $\operatorname{Sp}_{2,{\mathbbm R}}$ evolving under Eq.~(\ref{eq:system}) or Eq.~(\ref{eq:esystem})
can be re-parametrized in terms of the singular value decomposition parameters $\theta$, $\phi$ and $z$ such that 
$S=R_\theta Z R_\phi$. This re-parametrization can also be enacted for the function $f$ of Lemma \ref{lemmaf}, whence one obtains  
that the function
\begin{equation}
  f_z(\theta, \phi) =  \cos [2\theta] \cos [2\phi]-g(z, \phi)\sin [2\theta] \, ,
  \label{eq:reachregion}
\end{equation}
where 
 \begin{equation}
 g(z, \phi) := \frac{1}{2}\left(z^2+\frac{1}{z^2}\right) \sin [2\phi] - \frac{1}{2}\left(z^2-\frac{1}{z^2}\right),
\end{equation}
satisfies 
\begin{equation}\label{eq:fzbound}
  f_z( \theta,  \phi) \geq 1
\end{equation}
and
\begin{equation}\label{eq:fztbound}
\d{}{t}f_z( \theta,  \phi) \geq 0,
\end{equation}
for Eq.~(\ref{eq:esystem}).
The function $f_z$ is obtained as a change of coordinates from the function $f$ of Lemma \ref{lemmaf}, 
as shown explicitly in Appendix \ref{sec:calculations}.

\section{Unstable systems cannot reach passive operations}

We are now in a position to prove that single-mode unstable systems, in which no elliptic dynamical generators are enactable, do not allow passive operations to be reached. 

\begin{lemma}\label{thm:zalpha}
$f_z > f_1$ has solutions if and only if $z > 1$ and $\sin [2\theta] > 0$ for a given element of $\operatorname{Sp}_{2,\mathbbm R}$.
\end{lemma}
\begin{proof}
Consider the difference between $g(z, \phi)$ and $\sin [2\phi]$:
\begin{equation}
 \delta := g(z, \phi) - \sin [2\phi],
\end{equation}
so that
\begin{equation}
 f_z \equiv \cos[2(\theta- \phi)] - \delta\sin [2\theta].
\end{equation}
If $\sin [2\theta] < 0$ then $f_z > f_1$ if and only if $\delta > 0$. This is true if and only if $g(z,\phi) > \sin [2\phi]$. Hence
\begin{equation}
\frac{1}{2}\left(z^2+\frac{1}{z^2}\right) \sin [2\phi] - \frac{1}{2}\left(z^2-\frac{1}{z^2}\right) > \sin [2\phi],
\end{equation}
equally
\begin{equation}
 (z^2+\frac{1}{z^2}-2)\sin [2\phi] > z^2 - \frac{1}{z^2}. 
\end{equation}
$z^2+\frac{1}{z^2} - 2$ is positive for all values of $z$ and $\sin [2\phi] \leq 1$. Therefore this has solutions if and only if
\begin{equation}
 z^2+\frac{1}{z^2}-2 > z^2 - \frac{1}{z^2},
\end{equation}
which only has solutions for $z < 1$ which we have rule out in the definition of the unique singular value decomposition.

Conversely if $\sin [2\theta] > 0$ then $f_z > f_1$ if and only if $\delta < 0$. This is true if and only if $g(z,\phi) < \sin [2\phi]$. Hence
\begin{equation}
\frac{1}{2}\left(z^2+\frac{1}{z^2}\right) \sin [2\phi] - \frac{1}{2}\left(z^2-\frac{1}{z^2}\right) < \sin [2\phi],
\end{equation}
equally
\begin{equation}
 (z^2+\frac{1}{z^2}-2)\sin [2\phi] < z^2 - \frac{1}{z^2}. 
\end{equation}
$z^2+\frac{1}{z^2} - 2$ is positive for all values of $z$ and $\sin [2\phi] \geq -1$. Therefore this has solutions if and only if
\begin{equation}
 -(z^2+\frac{1}{z^2}-2) < z^2 - \frac{1}{z^2},
\end{equation}
which only has solutions for $z > 1$.
\end{proof}

\begin{lemma}\label{thm:zbound}
The existence of solutions for $f_z( \theta,\phi) > d$, where $d$ is some real number greater than or equal to one, implies that 
\begin{equation}
z > \sqrt{\frac{d+1}{2}}.
\end{equation}
\end{lemma}
\begin{proof}
Again consider:
\begin{equation}
 \delta := g(z,\phi) - \sin [2\phi]
\end{equation}
so that
\begin{equation}
 f_z \equiv \cos[2(\theta-\phi)] - \delta\sin [2\theta].
\end{equation}
The maximum value of $f_1$ is 1. From Lemma \ref{thm:zalpha} we know that we must select $\sin [2\theta] > 0$. There is a solution to the inequality $f_z > d$  if and only if there is a solution to
\begin{equation}
\delta < -(d-1).
\end{equation}
This translates to 
\begin{equation}
\frac{1}{2}\left(z^2+\frac{1}{z^2}\right) \sin [2\phi] - \frac{1}{2}\left(z^2-\frac{1}{z^2}\right) < \sin [2\phi] - (d-1),
\end{equation}
equally
\begin{equation}
 (z^2+\frac{1}{z^2}-2)\sin [2\phi] < z^2 - \frac{1}{z^2} - (d-1).
\end{equation}
$z^2+\frac{1}{z^2} - 2$ is positive for all values of $z$ and $\sin [2\phi] \leq 1$. Therefore this has solutions if and only if
\begin{equation}
 z^2+\frac{1}{z^2}-2 <  z^2 - \frac{1}{z^2} - (d-1),
\end{equation}
which only has solutions for
\begin{equation}
z^2 > \frac{d+1}{2}.
\end{equation}
$z > 1$ and so the statement is proven.
\end{proof}

\begin{lemma}\label{thm:noneutrals}
There does not exist $X \in \widetilde{\mathcal{R}}$ such that
\begin{equation}
X = SR_\theta S^{-1},
\end{equation}
where $S \in \operatorname{Sp}_{2,\mathbbm R}$, $R_\theta \in \operatorname{SO}(2)$.
\end{lemma}
\begin{proof}
Assume there exists $X \in \widetilde{\mathcal{R}}$ that satisfies the above condition. We may state that
\begin{equation}
X^m \in \widetilde{\mathcal{R}} \quad \forall m \in \mathbb{N}
\end{equation}
because the reachable set of Eq.~(\ref{eq:esystem}) has a semigroup structure. Note that
\begin{equation}
\begin{aligned}
 \lVert X^m-\mathbb{I} \rVert &= \lVert S(R_\theta^m - \mathbb{I})S^{-1} \rVert \\
&\leq \lVert S \rVert \lVert S^{-1} \rVert \lVert R_\theta^m-\mathbb{I} \rVert,
\end{aligned}
\end{equation}
where we use the Euclidean norm
\begin{equation}
 \lVert X \rVert := \sqrt{\Tr [X^{\sf T}X]}.
\end{equation}
$S$ is time-independent and so $||S||||S^{-1}||$ is constant. $R_\theta$ is quasi-periodic and so there must exist some $m$ such that
\begin{equation}
 \lVert R_\theta^m - \mathbb{I}\rVert < \varepsilon, \quad \forall \varepsilon > 0 
\end{equation}
and so there exists $m$ such that
\begin{equation}
\lVert X^m - \mathbb{I} \rVert < \varepsilon, \quad \forall \varepsilon > 0.
\end{equation}

From Eq.~(\ref{eq:fztbound}) we know that the value of $f_z$ must be non-decreasing along any trajectory of the system and from Eq.~(\ref{eq:fgradientA}) we see that its rate of change at identity is $1$. As a result, for some finite evolution time of Eq.~(\ref{eq:esystem}), all subsequent trajectories must contain elements that have a lower bound on their value of $f_z$ that is greater than $1$. By Lemma \ref{thm:zbound} this implies a lower bound on the value $z$ along a given trajectory of the control system given some minimal evolution time. $X^m$ is a possible trajectory of the system for all $m$ and from the above analysis we see that we can find $m$ such that the $z$ value of $X^m$ is arbitrarily close to $1$ violating the lower bound. Therefore $X$ cannot be an element of $\widetilde{\mathcal{R}}$.
\end{proof}

\begin{theorem}\label{thm:mainresult}
If Eq.~(\ref{eq:system}) is restricted to hyperbolic dynamical generators then its reachable set does not contain any elements of $\operatorname{SO}(2)$ except for $\mathbb{I}$.
\end{theorem}

\begin{proof}
The reachable set $\mathcal{R}$ of Eq.~(\ref{eq:system}) is symplectically similar to $\widetilde{\mathcal{R}}$. Lemma \ref{thm:noneutrals} states that $\widetilde{\mathcal{R}}$ does not contain any element that is symplectically similar to an element of $\operatorname{SO}(2) \setminus \mathbb{I}$. Thus $\mathcal{R}$ does not contain any element of $\operatorname{SO}(2) \setminus \mathbb{I}$.
\end{proof}

In practice our result implies that, given a set of unstable Hamiltonians, no manipulation in time of the control functions 
ever allows one to achieve any optical phase-shift operation. This holds even if the control Hamiltonians are 
able to generate the whole symplectic algebra.

\section{Example system: controlled squeezing Hamiltonians}\label{sec:example}

In the following we consider a specific example of an `unstable' system satisfying the Lie algebra rank criterion, and we 
study the corresponding reachable set.
In the single-mode scenario, such a system can be obtained by taking drift and control Hamiltonian 
as squeezing operations along different directions. Thus we consider the total Hamiltonian
\begin{equation}
\hat{H} = \hat{H}_A + u(t)\: \hat{H}_B,
\end{equation}
where the drift Hamiltonian $\hat{H}_A$ and the control Hamiltonian $\hat{H}_B$ are defined as 
\begin{equation}
\begin{aligned}\label{eq:exhams}
\hat{H}_A &= \frac{(1-c)\hat{x}^2 - (1+c)\hat{p}^2}{2}, \\
\hat{H}_B &= -\frac{\hat{x}\hat{p} + \hat{p}\hat{x}}{2},
\end{aligned}
\end{equation}
which correspond to Hamiltonian matrices
\begin{equation}\label{eq:exhammatrices}
 H_A = \begin{pmatrix} 1 - c & 0 \\ 0 & -c - 1 \end{pmatrix}, \quad
 H_B = \begin{pmatrix} 0 & -1 \\ -1 & 0 \end{pmatrix}.
\end{equation}
In terms of symplectic matrices this translates, as per Eq.~(\ref{eq:system}), into the open-loop control problem
\begin{equation}\label{eq:exsystem}
 \dot{S}(t) = (A+u(t)B)S(t),
\end{equation}
where
\begin{equation}
A \begin{pmatrix} 0 & -(1+c) \\ -(1-c) & 0 \end{pmatrix}, \quad B = \begin{pmatrix} -1 & 0 \\ 0 & 1 \end{pmatrix}.
 \label{eq:ctlrproblem}
\end{equation}
The reason for introducing this example system is to visualise its reachable set in order to give intuition for the behaviour of these systems in $n$ modes. 

Note that $A$ is parabolic, hyperbolic or elliptic if, respectively, $|c|$ is equal to, less than or greater than 1. We focus our analysis on $c=0$ to explore the hyperbolic case but in the next subsection remark on changes of behaviour as we vary $c$.

The visualisation will use the fact that the single-mode symplectic group is three dimensional and so a symplectic matrix in this group can be specified by three parameters. We choose the three parameters, $z$, $\theta$ and $\phi$ from the singular value decomposition and set them as axes on a graph. In order to plot the reachable set uniquely it is necessary to set $\theta_0$ and $\phi_0$, as discussed in Sec.~\ref{sec:svd}. We choose
\begin{equation}
 (\theta_0,\phi_0) = (0,\frac{\pi}{2}),
\end{equation}
which implies
\begin{equation}
 -\pi \leq \theta < \pi, \quad 0 \leq \phi < \pi.
 \label{eq:svdangranges}
\end{equation}
where we know that $z \geq 1$.

We will represent the reachable set as points in a cubic space with the $z$, $\theta$ and $\phi$ ranges as given. One may object that, although for $z>1$ this provides a one to one map between points on the plane and symplectic matrices, for $z=1$ this will not be the case because these matrices should be indicated by \textit{one} and not \textit{two} parameters, as per Eq.~(\ref{eq:svddecomposition}). However by Theorem \ref{thm:mainresult} we know that none of the elements of this plane will be reachable except for identity. After finding a point for identity we may maintain the cubic plot for illustrative clarity.
 
Analytically, we may find the `singular decomposition of identity' by considering the limit $t\to 0$ for any instance of reachable element. Take, for example,
\begin{equation}
  \exp\left[\begin{pmatrix} 0 & -1 \\ -1 & 0 \end{pmatrix} t\right]
\end{equation}
as $t \to 0$. Consider $t = \frac{1}{n}$, where $n \in \mathbb{N}$.
\begin{equation}
\begin{aligned}
  \exp\left[\begin{pmatrix} 0 & -1 \\ -1 & 0 \end{pmatrix} \frac{1}{n}\right] &= \left( R_{-\frac{3\pi}{4}} \begin{pmatrix} \frac{1}{e} & 0 \\ 0 & e \end{pmatrix} R_{\frac{3\pi}{4}} \right)^{\frac{1}{n}} \\
&= R_{-\frac{3\pi}{4}} \begin{pmatrix} \frac{1}{e} & 0 \\ 0 & e \end{pmatrix}^{\frac{1}{n}} R_{\frac{3\pi}{4}}.
\end{aligned}
\end{equation}
In the limit as $n \to \infty$ we find that the singular value decomposition of the identity is singled out as 
\begin{equation}\label{eq:identitydecomposition}
 \mathbb{I} = R_{-\frac{3\pi}{4}} R_{\frac{3\pi}{4}}.
\end{equation}
Using this result, we may also derive a bound on the angle $\theta$ that will appear in the numerics. Lemma \ref{thm:zalpha} states that for unstable systems $\sin[2\theta] > 0$ and so, given the range as set by Eq.~(\ref{eq:svdangranges}), this implies
\begin{equation}
-\pi < \theta < -\frac{\pi}{2}, \quad 0 < \theta < \frac{\pi}{2}.
\end{equation}
Eq.~(\ref{eq:identitydecomposition}) indicates that at $t = 0$, $\theta =-3\pi/4$. The singular value decomposition of elements must vary continuously and therefore
\begin{equation}
-\pi < \theta < -\frac{\pi}{2}.
\end{equation}

The numerics reported in the following subsection confirm and extend this analytical characterisation.

\subsection{Numerical study through optimal control}\label{numerics}

In this section we complement our analytics by applying optimal control algorithms adapted to the symplectic case in order to explore the reachable set of Eq.~(\ref{eq:exsystem}). We look to determine whether specific symplectic transformations $S_{\textit{target}}$ can be performed on our system given a fixed evolution time $T$. To test for controllability we implemented specific modules for simulating control in symplectic systems into into QuTiP, which is an open source python library for simulating quantum dynamics \cite{qutip1, qutip2}. The GRAPE algorithm \cite{grape} is used to attempt to find a control function $u(t)$ that will drive the system to perform the transformation $S_{\textit{target}}$. The evolution time $T$ is split into $Q$ equal time slices of length $\Delta t$ with the time at the beginning of each slice $t_k$. $u(t_k)$ is constant throughout the time slice, hence the piecewise constant control function $u(t)$ corresponds to a set of $Q$ real values. In this case $Q=10$.

The dynamical generators used in QuTiP are of the form
\begin{equation}
 H_k = H_A + u(t_k) H_B, \quad u(t_k) \in \mathbbm{R}.
\end{equation}
where $H_A$ and $H_B$ are as given in Eq.~(\ref{eq:exhammatrices}).

The evolution in each time slice is given by
\begin{equation}
 S_k = e^{\Omega H_k \Delta t}.
\end{equation}
The full evolution is given by
\begin{equation}
 S(T) = S_Q S_{Q-1} \cdots S_k \cdots S_2 S_1.
\end{equation}

The difference between the evolved transformation and the target is quantified by the \emph{fidelity error} (or \emph{infidelity}) as measured by the Frobenius norm
\begin{equation}
 \varepsilon := \lambda \Tr[(S(T) - S_{\textit{target}})^{\sf T} (S(T) - S_{\textit{target}})]\,,
\end{equation}
with $\lambda = 1/8$ for a $2 \times 2$ matrix.

The control function is optimised to minimise $\varepsilon$ using the L-BFGS-B method in the scipy optimization function, which is a wrapper to the implementation by Byrd \textit{et al}. \cite{lbfgsb}. The exact gradient with respect to $u(t_k)$ is calculated using the Frechet derivative (or augmented matrix method) as described in Eq.~(12) of \citep{robustqgates}. The target is considered achieved in this case if $\varepsilon < 10^{-3}$. The control function optimisation terminates unsuccessfully if either a local minima is found or a processing time limit is exceeded.

The set of possible target symplectics is discretised in the $(z,\theta,\phi)$ space by only considering points at $\pi/12$ intervals in the angular directions and $10$ logarithmically equal intervals between $z=1$ and the arbitrary upper bound of $z=100$. A bisection method was used to determine the boundary between reachable and unreachable targets. The boundary points are depicted as the darker blue points in Fig.~(\ref{fig:reachsetplots}). These simulations were repeated for combinations of $c = \{0.0, \pm0.5, \pm0.9, \pm0.99, \pm1.01, \pm1.1, \pm1.5\}$ and $T = \{0.1, 0.5, 1, 2, 3, 4, 5, 7, 10, 20, 50, 100\}$. When successful, the $S_{\textit{target}}$ test shows that there is at least one set of controls that can achieve the target transformation.

\begin{figure}
\centering
\raggedright (a) $c=0, T=1$\\
\begin{center}
\includegraphics[width=0.67\linewidth]{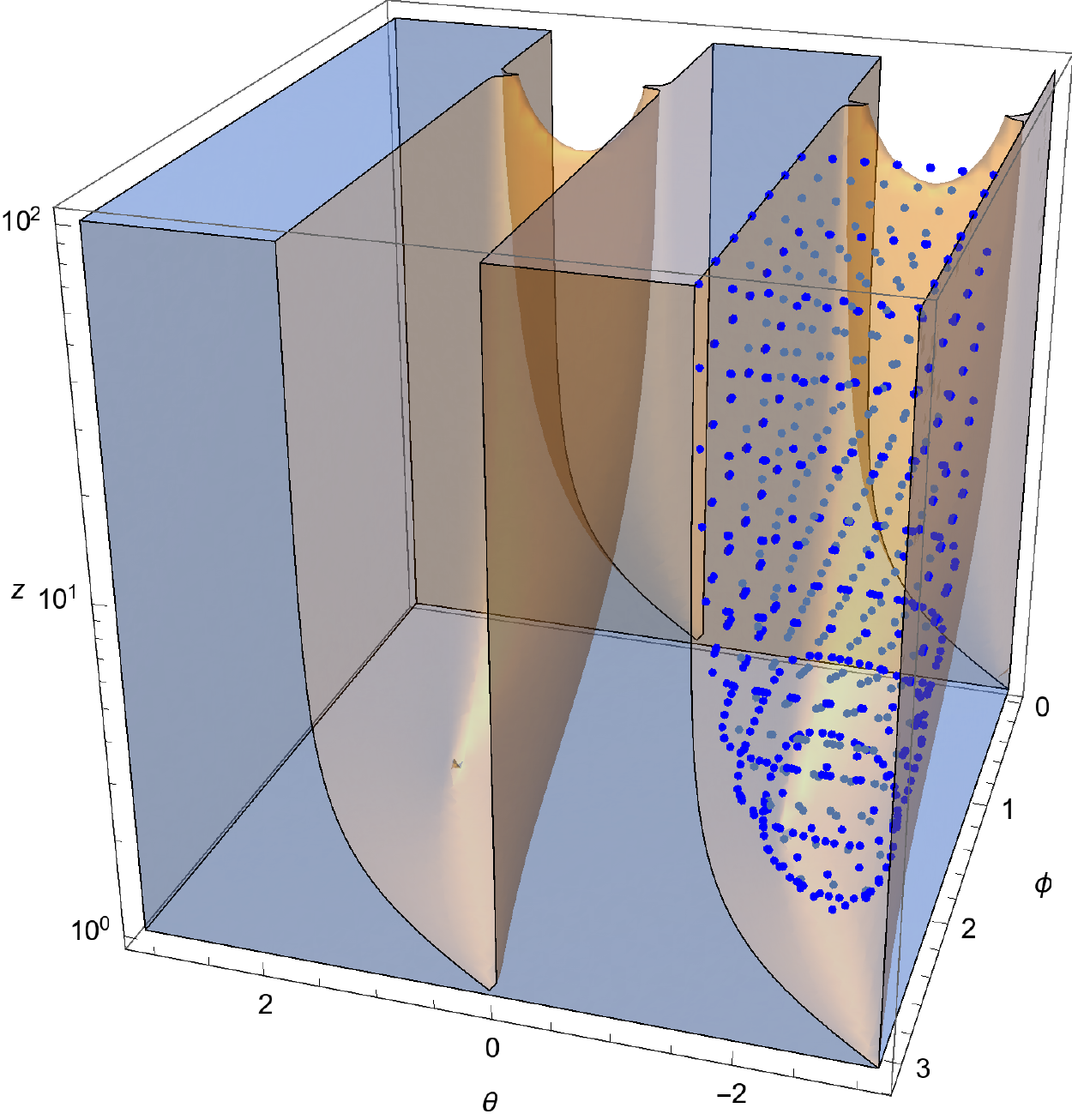}\\
\end{center}
\raggedright (b) $c=0, T=5$\\
\begin{center}
\includegraphics[width=0.67\linewidth]{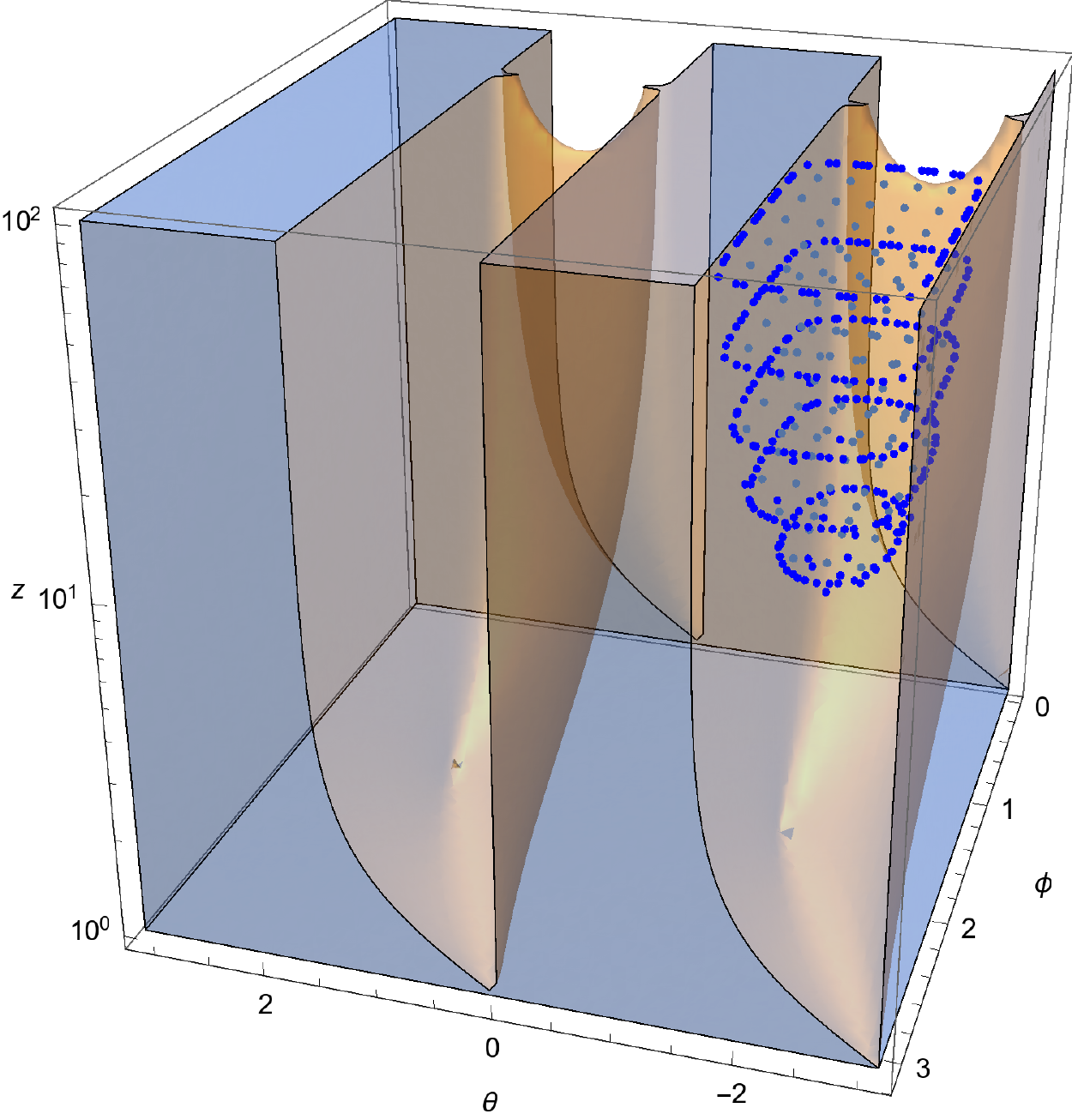}\\
\end{center}
\raggedright (c) $c=-0.99, T=5$\\
\begin{center}
\includegraphics[width=0.67\linewidth]{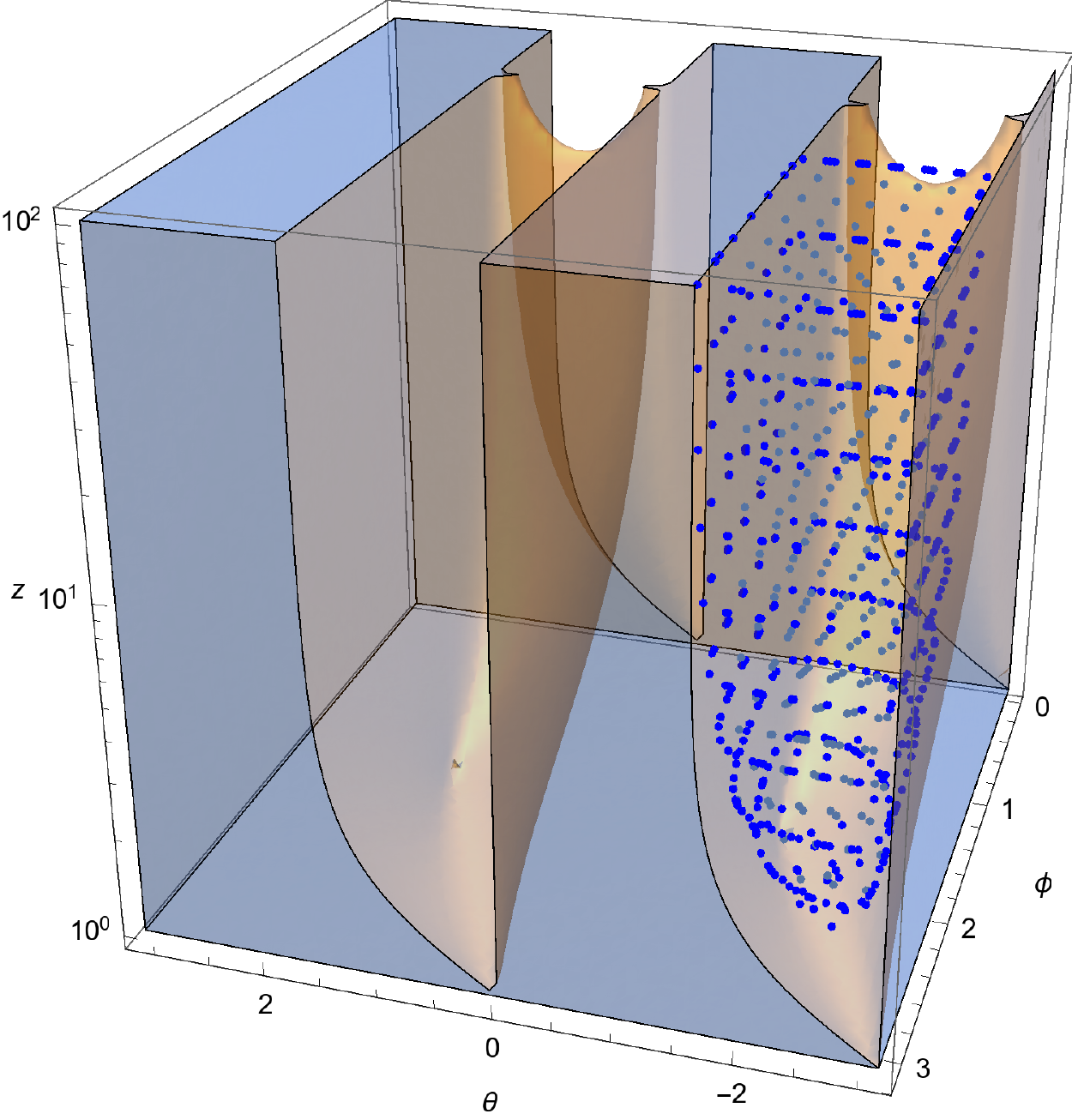}\\
\end{center}
\caption{The blue points are reachable operations in the SVD basis after time $T$. The enclosed blue region is unreachable by Eq.~(\ref{eq:fzbound}). Specific parameters: (a) $c=0, T=1$; (b) $c=0, T=5$; (c) $c=-0.99, T=5$.}
\label{fig:reachsetplots}
\end{figure}

The results of some of the tests for unstable systems are shown in Fig.~(\ref{fig:reachsetplots}). Note that the points shown are those reachable specifically \textit{at} evolution time $T$ rather than up to time $T$.
For unstable systems the reachable points are restricted to a set centred around $(\theta,\phi) = (-3\pi/4, 3\pi/4)$ and bounded by $-\pi < \theta < -\pi/2$, confirming the analytics, and $\pi/2 < \phi < \pi$, which was not proved analytically. This indicates that the numerics supply a tighter bound than the analytics.

The example system is demonstrated to be unstable for $-1<c<1$. For $|c| \geq 1.1$ all points were found to be reachable. For $|c| = 1.01$ the optimiser was unable to find a suitable control function for some $S_{\textit{target}}$. These unreached points were predominantly in the region found reachable for $-1<c<1$. However, it is most likely that this is due to the constraints placed on the pulse optimisation, and demonstrates the difficulty of finding a solution near the edge of stability. Fig.~(\ref{fig:reachsetplots}) shows the case for $c = -0.99$ where we see that the reachable set is broader. There is then a discontinuity as we pass $|c| = 1$ when the reachable set then becomes the whole space. The broadening of the reachable set as $c$ goes near the boundary indicates that the control system has become in a sense more `stable'.

The numerics show that, in one mode, when an elliptic drift field cannot be constructed, the system will restrict itself to unbounded squeezing within a small angular region. The ability to visualise this behaviour is by virtue of working in a single mode and a generalisation of this would require a more sophisticated treatment. Nevertheless, working on the numerics for this case provides some much needed intuition for a higher mode exploration.

\section{Conclusions and outlook}\label{sec:conclusion}
The Lie algebra rank criterion is necessary and sufficient for the controllability of systems evolving under compact Lie groups. The symplectic group $\operatorname{Sp}_{2n,\mathbbm R}$ is non-compact and so the criterion loses its sufficiency for these systems. The discovery of a necessary and sufficient condition for control on the symplectic group is still an open problem in mathematical control theory \cite{ElAssoudi2014}.

In this paper, we sought to characterise and visualise single-mode systems that obey the Lie algebra rank criterion but are not controllable, as 
well as to connect the mathematical treatment with the formalism and terminology of quantum optics. 
In this process we found that such systems are unable to reach non-trivial elements of the maximal compact subgroup of the symplectic group. 
Mathematically, this provides a new characterisation of an intriguing situation that has never been systematically analysed.  
In terms of physics these new results translate into the inability to enact passive operations -- phase shifters for one optical mode -- given only unstable Hamiltonians. 
Note that controlled operations generated by unstable generators 
are accessible in several experimental set-ups, both optical and mechanical 
and, given the exponential speed-up they grant, are instrumental in beating decoherence times 
-- see, for instance, reference \cite{oriol}, where such operations are proposed to achieve this aim in the context of 
superconducting quantum magnetomechanics, and the Hamiltonians generating them are referred to as 
``repulsive potentials''.

It is expected that this characterisation will extend to $n$ modes and provide physical insight into the existing mathematical and experimental problems surrounding the control of closed quantum systems.

\section{Acknowledgements}
We thank Ugo Boscain, Mario Sigalotti and Pierre Rouchon for discussions 
(during which Ugo pointed us to reference \cite{Wu2007}).
MGG and AS acknowledge financial support financial support
from EPSRC through grant EP/K026267/1. 
MGG acknowledges
support from the Marie Sk\l odowska-Curie Action H2020-MSCA-IF-2015. 
DB acknowledges support from EPSRC grant EP/M01634X/1.
We are grateful to HPC Wales for giving access to the cluster that was used to perform the numerical simulations.

\appendix

\section{Complete proof of the uncontrollability condition}\label{sec:proofs}

In order to prove Lemmata \ref{thm:symplecticsimilarity} and \ref{lemmaf}, we need a few preliminary statements, 
which are also taken directly from \cite{Wu2007}.

\begin{lemma}\label{thm:parablemma}
The equation
\begin{equation}
 \Tr[[M,N]^2] = \Tr[MN]^2-2\Tr[N^2]\Tr[M^2]
\end{equation}
holds for $M,N \in \mathfrak{sp}_{2,\mathbbm R}$.
\end{lemma}
\begin{proof}
First we expand the elements in the basis defined in Eq.~(\ref{eq:basis}):
\begin{equation}\label{eq:mdecomp}
M = m_1K_x + m_2K_y + m_3K_z,
\end{equation}
\begin{equation}\label{eq:ndecomp}
N = n_1K_x + n_2K_y + n_3K_z,
\end{equation}
\begin{equation}\label{eq:mndecomp}
\begin{aligned}
\lbrack M, N \rbrack &= (m_2n_3-m_3n_2)K_x \\
& \quad+ (m_3n_1-m_1n_3)K_y \\
& \quad-(m_1n_2-m_2n_1)K_z.
\end{aligned}
\end{equation}
We use this expansion to express the value of the following terms:
\begin{equation}\label{eq:trm2}
\Tr[M^2] = \frac{1}{2}(m_1^2+m_2^2-m_3^2),
\end{equation}
\begin{equation}
\Tr[N^2] = \frac{1}{2}(n_1^2+n_2^2-n_3^2),
\end{equation}
\begin{equation}
\Tr[MN] = \frac{1}{2}(m_1n_1+m_2n_2-m_3n_3),
\end{equation}
\begin{equation}\label{eq:trmn2}
\begin{aligned}
\Tr[[M,N]^2] = &\frac{1}{2}((m_2n_3-m_3n_2)^2 \\&+ (m_3n_1-m_1n_3)^2 \\&- (m_1n_2-m_2n_1)^2).
\end{aligned}
\end{equation}
Then we combine them to prove the statement:
\begin{equation}
 \Tr[[M,N]^2] = \Tr[MN]^2-2\Tr[N^2]\Tr[M^2].
\end{equation}
\end{proof}

\begin{lemma}\label{thm:parabthing}
 If $\Tr([A,B]^2) = 0$ in Eq.~(\ref{eq:system}) then the system does not obey the Lie algebra rank criterion.
\end{lemma}
\begin{proof}
From Eqs.~(\ref{eq:mdecomp}), (\ref{eq:ndecomp}) and (\ref{eq:mndecomp}) it can be concluded that $M$, $N$ and $[M,N]$ are linearly dependent if and only if 
\begin{equation}
\det\begin{pmatrix} 
m_1 & n_1 & m_2n_3-m_3n_2 \\
m_2 & n_2 & m_3n_1-m_1n_3 \\ 
m_3 & n_3 & -(m_1n_2-m_2n_1)
\end{pmatrix} = 0,
\end{equation}
or equivalently
\begin{equation}
(m_2n_3-m_3n_2)^2 + (m_3n_1-m_1n_3)^2 - (m_1n_2-m_2n_1)^2 = 0.
\end{equation}
From Eq.~(\ref{eq:trmn2}) we see that this is equivalent to $[M,N]$ being parabolic. If $A$, $B$ and $[A,B]$ are linearly dependent then the span of $A$ and $B$ does not generate $\mathfrak{sp}_{2,\mathbbm R}$.
\end{proof}

\begin{lemma}\label{thm:hypbconversion}
Consider hyperbolic $M \in \mathfrak{sp}_{2,\mathbbm R}$. There exists $P \in \operatorname{Sp}_{2,\mathbbm R}$ such that $PMP^{-1} = \sqrt{2\Tr[M^2]}K_y$.
\end{lemma}
\begin{proof}
 Expand $M$ so we get
\begin{equation}
 M = xK_x+yK_y+zK_z,
\end{equation}
where $\Tr[M^2] = \frac{1}{2}(x^2+y^2-z^2) > 0$. First we seek a matrix $P_1 = e^{ \alpha K_z} \in \operatorname{Sp}_{2,\mathbbm R}$ which satisfies
\begin{equation}
 P_1MP_1^{-1} = \sqrt{x^2+y^2} K_y + z K_z.
\end{equation}
Let $ \alpha$ be the angle satisfying
\begin{equation}
\sin [\alpha] = \frac{x}{\sqrt{x^2+y^2}}, \quad \cos [\alpha] = \frac{y}{\sqrt{x^2+y^2}}.
\end{equation}
According to the formula 
\begin{equation}\label{eq:notbch}
e^M N e^{-M} = N + [M,N] + \frac{1}{2!}[M,[M,N]] + \ldots,
\end{equation}
one can immediately obtain that
\begin{widetext}
\begin{equation}
\begin{aligned}
e^{ \alpha K_z} M e^{- \alpha K_z} &= xe^{ \alpha K_z} K_x e^{- \alpha K_z} + ye^{ \alpha K_z} K_y e^{- \alpha K_z} + zK_z \\
&= (x\cos [\alpha]-y\sin [\alpha])K_x + (x\sin [\alpha] + y\cos [\alpha])K_y + zK_z \\
&= \sqrt{x^2+y^2} K_y + zK_z.
\end{aligned}
\end{equation}

Next we show that there is a matrix $P_2 = e^{ \beta K_x} \in \operatorname{Sp}_{2,\mathbbm R}$ which can convert $\sqrt{x^2+y^2}K_y + zK_z$ into $\sqrt{2\Tr[M^2]}K_y$. Since $x^2+y^2-z^2 > 0$ we can choose $ \beta$ such that
\begin{equation}
\sinh [\beta] = \frac{z}{\sqrt{x^2+y^2-z^2}}, \quad \cosh [\beta] = \frac{\sqrt{x^2+y^2}}{\sqrt{x^2+y^2-z^2}}.
\end{equation}
Make use of Eq.~(\ref{eq:notbch}) again and obtain
\begin{equation}
\begin{aligned}
e^{ \beta K_x}(\sqrt{x^2+y^2}K_y+zK_z)e^{- \beta K_x} &= \sqrt{x^2+y^2}e^{ \beta K_x} K_y e^{- \beta K_x} + ze^{ \beta K_x} K_z e^{- \beta K_x} \\
&= (\sqrt{x^2+y^2}\cosh [\beta]-z\sinh [\beta])K_y + (z\cosh [\beta]-\sqrt{x^2+y^2}\sinh [\beta])K_z \\
&= \sqrt{x^2+y^2-z^2}K_y \\
&= \sqrt{2\Tr[M^2]}K_y.
\end{aligned}
\end{equation}
\end{widetext}
Consequently the $\operatorname{Sp}_{2,\mathbbm R}$ matrix $e^{ \beta K_x}e^{ \alpha K_z}$ will convert $M$ into $\sqrt{2\Tr[M^2]}K_y$ when $M$ is hyperbolic.
\end{proof}

We now proceed to provide a proof for Lemma \ref{thm:symplecticsimilarity}. First we restate it.
\begin{thm-hand}[1]
 If $\Xi$ only contains hyperbolic elements then Eq.~(\ref{eq:system}) is similar, via a symplectic transformation, to 
\begin{equation}\label{eq:esystemA}
 \dot{S}(t) = (-K_x + bK_z + u(t)K_y)S(t), \quad S(0) = \mathbb{I}_2,
\end{equation}
where $b$ is some real constant with modulus strictly less than one.
\end{thm-hand}

\begin{proof}
If Eq.~(\ref{eq:system}) only has hyperbolic controls then the following inequality holds:
 \begin{equation}
  \Tr[(A+vB)^2] = \Tr[B^2]v^2+2\Tr[AB]v+\Tr[A^2] > 0,
 \end{equation}
for all $v \in \mathbbm{R}$. For this inequality to hold for all $v$ it is immediately clear that $\Tr[A^2] > 0$. We can see that $\Tr[B^2] > 0$ because (a) if it were less than zero then there exists $v$ for which the inequality does not hold and (b) if it were equal to zero then $\Tr[AB]$ must equal zero; by Lemma \ref{thm:parablemma} this implies that $\Tr([A,B]^2) = 0$ which implies that the system does not obey the Lie algebra rank criterion by Lemma \ref{thm:parabthing} which would contradict our assumption. 
 
With the knowledge that $B$ is hyperbolic, Lemma \ref{thm:hypbconversion} states that there exists a symplectic similarity transformation to transform Eq.~(\ref{eq:system}) into:
\begin{equation}\label{eq:eesystemA}
\dot{S}(t) = (A' + u(t)K_y)S(t), \quad S(0) = \mathbb{I}_2,
\end{equation}
where $A'$ is some unspecified element of $\mathfrak{sp}_{2,\mathbbm R}$. Expand $A'$ in the symplectic basis of Eq.~(\ref{eq:basis}):
\begin{equation}
A' = b_xK_x + b_yK_y + b_zK_z.
\end{equation}
By redefining $u(t)$ we can transform the system such that $b_y$ equals zero. We know that $A'$ is hyperbolic because this property is invariant under similarity transformation, therefore we know that $|b_x| > |b_z|$ from Eq.~(\ref{eq:trm2}). The role of time in Eq.~(\ref{eq:eesystemA}) allows us to rescale such that the coefficient of $K_x$ has modulus one leaving us with system
\begin{equation}
\dot{S}(t) = (\epsilon K_x + bK_z + u(t)K_y)S(t), \quad S(0) = \mathbb{I}_2,
\end{equation}
where $|b| < 1$ and $\epsilon = \pm1$. If $\epsilon = -1$ then we leave the system as it is and the proof is finished. If $\epsilon = 1$ then enacting a similarity transformation under the symplectic matrix $\Omega$ is equivalent to time reversal and sends each of the basis matrices to their negative. Thus we have shown that Eq.~(\ref{eq:system}) is symplectically similar to Eq.~(\ref{eq:esystemA}). Note that we did not talk about effects on the initial value of $X$ because this is \textit{set} to be $\mathbb{I}_2$.
\end{proof}

We now proceed to provide a proof for Lemma \ref{lemmaf}. First we restate it.
\begin{thm-hand}[2]
Any real $2 \times 2$ matrix can be written as
\begin{equation}\label{eq:xgeneralA}
X = \begin{pmatrix} x_1+x_3 & x_2+x_4 \\ x_4-x_2 & x_1-x_3 \end{pmatrix},
\end{equation}
where $x_i \in {\mathbbm R}$. If $X \in \widetilde{\mathcal{R}}$ then the function
\begin{equation}
  f(x_1,x_2,x_3,x_4) := (x_1-x_4)^2 - (x_2-x_3)^2
\end{equation}
satisfies 
\begin{equation}
 f(x_1,x_2,x_3,x_4) \geq 1
\end{equation}
and
\begin{equation}
\d{}{t}f(x_1,x_2,x_3,x_4) \geq 0,
\end{equation}
for any choice of $u(t)$ in Eq.~(\ref{eq:esystem}), or equally Eq.~(\ref{eq:esystemA}). 
\end{thm-hand}
\begin{proof}
Eqs.~(\ref{eq:esystemA}) and (\ref{eq:xgeneralA}) provide the set of equations
\begin{align}
\dot{x}_1 &= \frac{1}{2}(ax_2-x_4-vx_3), \label{eq:elementtime1}\\
\dot{x}_2 &= \frac{1}{2}(-ax_1+x_3-vx_4), \label{eq:elementtime2}\\
\dot{x}_3 &= \frac{1}{2}(-ax_4+x_2-vx_1), \label{eq:elementtime3}\\
\dot{x}_4 &= \frac{1}{2}(ax_3-x_1-vx_2). \label{eq:elementtime4}
\end{align}
Subtracting Eqs.~(\ref{eq:elementtime1}) and (\ref{eq:elementtime4}) then followed by a succeeding multiplication by $2(x_1-x_4)$ provides
\begin{equation}\label{eq:elementtimex}
\begin{aligned}
\d{}{t} (x_1-x_4)^2 &= a(x_1-x_4)(x_2-x_3)\\ 
&+ (x_1-x_4)^2 \\
&+ v(x_1-x_4)(x_2-x_3).
\end{aligned}
\end{equation}
Similarly, we have 
\begin{equation}\label{eq:elementtimey}
\begin{aligned}
\d{}{t} (x_2-x_3)^2 = & -a(x_1-x_4)(x_2-x_3) \\
& - (x_2-x_3)^2 \\
&+ v(x_1-x_4)(x_2-x_3).
\end{aligned}
\end{equation}

Then subtracting Eqs.~(\ref{eq:elementtimex}) and (\ref{eq:elementtimey})
\begin{equation}\label{eq:fgradientA}
\begin{aligned}
\d{}{t} &\left( (x_1-x_4)^2-(x_2-x_3)^2\right) \\ 
&= \; 2a(x_1-x_4)(x_2-x_3) \\
&\qquad + \left((x_1-x_4)^2 + (x_2-x_3)^2\right) \\
&= \; (1-|a|)\left((x_1-x_4)^2 + (x_2-x_3)^2\right) 
\\&\qquad + |a|\left((x_1-x_4) -sign(a)(x_2-x_3)\right)^2 \\
&\geq 0 \, .
\end{aligned}
\end{equation}
Thus, the function $f$ is nondecreasing for every trajectory of the system. Since the initial value of $f$ is 1 it can be concluded that the reachable states of Eq.~(\ref{eq:esystemA}) should satisfy the restriction that $f \geq 1$.
\end{proof}

\section{Singular value decomposition}\label{sec:appsvd}
\subsection{Uniqueness of the singular value decomposition}\label{sec:uniquesvd}
To prevent any ambiguity we require that the singular value decomposition be unique. This is not true in general and therefore we need to restrict the range of allowed angles so that it is properly defined. In short, we want 
\begin{equation}
S = R_\theta Z R_\phi = R_\alpha Z' R_\beta
\end{equation}
to imply that $ \alpha = \theta$, $ \beta = \phi$ and $Z' = Z$. The first thing to notice is that the singular values of $S$ are unique and so we would only ever get either $Z' = Z$ or $Z' = Z^{-1}$. The latter case corresponds to the situation where $z < 1$ which may be ignored provided the range of the angles is properly limited allowing $Z^{-1} = R_{-\pi/2}ZR_{\pi/2}$. Thus we need only consider two cases, $z = 1$ and $z > 1$. In the conclusion we use these cases to show that we have a freedom in how to represent the singular value decomposition.

\subsubsection{$Z \neq \mathbbm{I}$}
Let's first look at the former case, $Z'=Z$, where $Z \neq \mathbb{I}$. Assume a non-unique decomposition:
\begin{equation}
 R_\theta Z R_\phi = R_ \alpha Z R_ \beta,
\end{equation}
or equivalently
\begin{equation}
R_{\theta- \alpha} Z = Z R_{ \beta-\phi},
\end{equation}
and explicitly
\begin{equation}
\begin{aligned}
 &\begin{pmatrix} \frac{1}{z}\cos[\theta- \alpha] & -z\sin[\theta- \alpha] \\ \frac{1}{z}\sin[\theta- \alpha] & z\cos[\theta- \alpha] \end{pmatrix} = \\&\begin{pmatrix} \frac{1}{z}\cos[ \beta-\phi] & -\frac{1}{z}\sin[ \beta-\phi] \\ z\sin[ \beta-\phi] & z\cos[ \beta-\phi] \end{pmatrix}.
\end{aligned}
\end{equation}
This implies the set of conditions 
\begin{equation}
\begin{aligned}
\frac{1}{z}\sin[\theta- \alpha] &= z\sin[ \beta-\phi], \\ 
z\sin[\theta- \alpha] &= \frac{1}{z}\sin[ \beta-\phi], \\
\cos[\theta- \alpha] &= \cos[ \beta-\phi],
\end{aligned}
\end{equation}
which only hold when
\begin{equation}
\begin{aligned}
 \sin[\theta -  \alpha] &= 0, \\
\sin[ \beta-\phi] &= 0, \\
\cos[\theta- \alpha] &= \cos[ \beta-\phi].
\end{aligned}
\end{equation}
These only hold when
\begin{equation}\label{eq:avoid}
  \alpha = \theta + n\pi \quad \text{and} \quad  \beta = \phi +m\pi
\end{equation}
for $n,m \in \mathbb{Z}$ either both odd or both even.

To avoid Eq.~(\ref{eq:avoid}) being satisfied for $m,n \neq 0$ we limit $\phi$ to vary in a range less than $\pi$ so that $\beta = \phi$. This sets $m=0$ and so to satisfy Eq.~(\ref{eq:avoid}) without letting $\alpha = \theta$ the nearest option would be to let $\alpha = \theta \pm 2\pi$. The maximum range for the angles governing $\operatorname{SO}(2)$ is $2\pi$ and so this is the bound that will apply to $\theta$. For uniqueness, therefore, we set the ranges of $\theta$ and $\phi$ to:
\begin{equation}
 -\pi+\theta_0 \leq \theta < \pi+\theta_0, \quad -\frac{\pi}{2}+\phi_0 \leq \phi < \frac{\pi}{2}+\phi_0,
\end{equation}
where $\theta_0$, $\phi_0$ fix the centre of the ranges.

\subsubsection{$Z = \mathbbm{I}$}
In this case we consider $Z = \mathbb{I}$. We look for times when
\begin{equation}
 R_\theta R_\phi = R_ \alpha R_ \beta
\end{equation}
is satisfied.

These are cases when
\begin{equation}
 \theta + \phi =  \alpha +  \beta +2n\pi,
\end{equation}
for $n \in \mathbb{Z}$.

This holds true for a whole range of angles. We can arbitrarily set $\phi = \phi_0$ to let $\theta$ label the elements of $\operatorname{SO}(2)$.

\subsubsection{Angle limit}
Now we have choices on how to set the angles such that the decomposition is unique. We choose
\begin{equation}
 -\pi+\theta_0 \leq \theta < \pi+\theta_0, \quad -\frac{\pi}{2}+\phi_0 \leq \phi < \frac{\pi}{2}+\phi_0
\end{equation}
to make the singular value decomposition unique when $Z \neq \mathbb{I}$. $\theta_0$ and $\phi_0$ are some constants that we are free to set. Note that we have made a further arbitrary choice in exactly where to make the bounds tight. For $Z = \mathbb{I}$ we must totally restrict one of the angles and leave the other free; we choose so set $\phi=\phi_0$.

\subsection{`Singular value decomposition' coordinates for $f$}\label{sec:calculations}
In this section we represent $\cos[\theta]$ as $c\theta$ and $\sin[\theta]$ as $s\theta$ for brevity. We begin with two expressions for $X \in \operatorname{Sp}_{2,\mathbbm R}$:
\begin{equation}
 X = \begin{pmatrix} x_1+x_3 & x_2+x_4 \\ x_4-x_2 & x_1-x_3 \end{pmatrix},
 \end{equation}
 and
\begin{equation}
 X = \begin{pmatrix} \frac{c\theta c\phi}{z} - zs\theta s\phi & -\frac{c\theta s\phi}{z}-zs\theta c\phi \\ \frac{s\theta c\phi}{z} + zc\theta s\phi & -\frac{s\theta s\phi}{z}+zc\theta c\phi \end{pmatrix}.
\end{equation}

Equating the two expression and solving for $x_i$ we find that
\begin{align}
 2x_1 &= \frac{1}{z}(c\theta c\phi - s\theta s\phi) + z(c\theta c\phi - s\theta s\phi), \\
 2x_2 &= -\frac{1}{z}(s\theta c\phi + c\theta s\phi) - z(s\theta c\phi + c\theta s\phi), \\
 2x_3 &= \frac{1}{z}(c\theta c\phi + s\theta s\phi) - z(s\theta s\phi + c\theta c\phi), \\
 2x_4 &= \frac{1}{z}(s\theta c\phi - c\theta s\phi) + z(c\theta s\phi - s\theta c\phi),
\end{align}
and so
\begin{widetext}
\begin{align}
 2(x_1-x_4) &= \frac{1}{z}(c\theta c\phi - s\theta s\phi - s\theta c\phi + c\theta s\phi) + z(c\theta c\phi - s\theta s\phi - c\theta s\phi + s\theta c\phi), \\
 2(x_2-x_3) &= -\frac{1}{z}(s\theta c\phi + c\theta s\phi + c\theta c\phi + s\theta s\phi) - z(s\theta c\phi + c\theta s\phi - s\theta s\phi - c\theta c\phi),
 \end{align}
or more simply
\begin{equation}
\begin{aligned}
 2(x_1-x_4) &= \frac{1}{z}(c\theta-s\theta)(c\phi+s\phi) + z(c\theta+s\theta)(c\phi-s\phi), \\
 2(x_2-x_3) &= -\frac{1}{z}(c\theta + s\theta)(c\phi + s\phi) - z(c\theta-s\theta)(-c\phi+s\phi),
 \end{aligned}
\end{equation}
which leads to
\begin{equation}
\begin{aligned}
 (x_1-x_4)^2 &= \frac{1}{4}\bigg(\frac{1}{z^2}(c\theta-s\theta)^2(c\phi+s\phi)^2+z^2(c\theta+s\theta)^2(c\phi-s\phi)^2 + 2(c\theta+s\theta)(c\theta-s\theta)(c\phi+s\phi)(c\phi-s\phi)\bigg), \\
 (x_2-x_3)^2 &= \frac{1}{4}\bigg(\frac{1}{z^2}(c\theta+s\theta)^2(c\phi+s\phi)^2+z^2(c\theta-s\theta)^2(c\phi-s\phi)^2-2(c\theta+s\theta)(c\theta-s\theta)(c\phi+s\phi)(c\phi-s\phi)\bigg).
 \end{aligned}
\end{equation}
Subtracting the two
\begin{equation}
\begin{aligned}
(x_1-x_4)^2 - (x_2-x_3)^2 = \frac{1}{4}\bigg(&\frac{1}{z^2}(c\phi+s\phi)^2\big((c\theta-s\theta)^2-(c\theta+s\theta)^2\big) + \\
&z^2(c\phi-s\phi)^2\big((c\theta+s\theta)^2-(c\theta-s\theta)^2\big)+ \\
&4(c\theta^2-s\theta^2)(c\phi^2-s\phi^2)\bigg),
\end{aligned}
\end{equation}
to
\begin{equation}
\begin{aligned}
(x_1-x_4)^2 - (x_2-x_3)^2 = \frac{1}{4}\bigg(\frac{1}{z^2}(1+s2\phi)(-s2\theta) +
z^2(1-s2\phi)(s2\theta)+4c2\theta c2\phi\bigg),
\end{aligned}
\end{equation}
to
\begin{equation}  
 (x_1-x_4)^2 - (x_2-x_3)^2 = c2\theta c2\phi - s2\theta \left(\frac{1}{2}\left(z^2+\frac{1}{z^2}\right) s2\phi - \frac{1}{2}\left(z^2-\frac{1}{z^2}\right)\right).
\end{equation}
which is our new expression for $f$ in terms of $\theta$, $\phi$ and $z$.
\end{widetext}


\begin{thebibliography}{99}
\bibitem{paranote}{Note that, here, the explicit reference to squeezing
is meant to exclude a specific subset of Hamiltonians which will be termed `parabolic' in the following. 
Such Hamiltonians, including notably the free-particle Hamiltonian $\hat{p}^2$, 
are bounded from below but do not show dynamical recurrence. As will be explained later, parabolic Hamiltonians 
do not give rise to situations where the rank criterion is satisfied and yet the system is not controllable,
and will hence be disregarded.}
\end{thebibliography}
\end{document}